\documentclass[sigconf]{acmart}
\AtBeginDocument{%
  }
\setcopyright{acmcopyright}
\copyrightyear{2025}
\acmYear{2025}

\acmConference[]{}{}{}
\acmBooktitle{}
\acmPrice{}
\acmISBN{}

\usepackage{amsthm,amsmath}
\usepackage{dsfont}
\usepackage{tikz}
\usepackage{todonotes}
\usepackage{hyperref}
\usepackage{nicefrac}
\usepackage{enumitem}
\usepackage{hyperref}
\usepackage[linesnumbered,lined,boxed]{algorithm2e}

\floatstyle{ruled}
\newfloat{algo}{tp}{lop}
\floatname{algo}{Algorithm}
\usepackage[noend]{algpseudocode}
\algrenewcommand\textproc{\sffamily}
\algrenewcommand\algorithmicindent{1em}
\newtheorem{theorem}{Theorem}[section]
\newtheorem{prop}[theorem]{Proposition}

\newtheorem{corollary}[theorem]{Corollary}
\newtheorem{lemma}[theorem]{Lemma}
\newtheorem{remark}[theorem]{Remark}

\newtheorem{defi}[theorem]{Definition}
\newtheorem{example}[theorem]{Example}

\newtheorem{hypo}[theorem]{Hypothesis}
\newtheorem{notation}[theorem]{Notation}

\def\eatspace#1{#1}
\def\step#1#2{\par\kern1pt\hangindent#2em\hangafter=1\noindent\rlap{\small#1}\kern#2em\relax\eatspace}

\def\<#1>{\langle#1\rangle}

\def\ud{\mathrm{d}}
\def\Fbb{\mathbb{F}}
\def\Dcal{\mathcal{D}}

\def\nicefrac#1#2{#1/#2}

\def\ord{\operatorname{ord}}

\begin{document}

\title{$\mathrm{LCM}$ decomposition of linear differential operators in positive
characteristic}
\thanks{%
  R.\ Pages was supported by the Austrian FWF grants 10.55776/PAT8258123. 
}

\author{Raphaël Pagès}
\email{raphael.pages@jku.at}
\orcid{0009-0001-9840-0680}
\affiliation{%
  \institution{Institute for Algebra}
  \institution{Johannes Kepler University}
  \state{}
  \city{4040 Linz}
  \country{Austria}
}

\renewcommand{\shortauthors}{R. Pagès}

\begin{abstract}
  We present an algorithm to compute $\mathrm{LCLM}$-decompositions for
  linear differentials operators with coefficients in the rational function
  field of characteristic $p$, $\mathbb{F}_{p^n}(t)$. We show that for an
  operator $L$ of order $r$ with coefficients of degree $d$, it finishes
  in polynomial time in $r$, $d$ and $p$. This algorithm proceeds in three
  steps. We begin by showing that the ``type'' of the factorisation of $L$ can be easily obtained from the
  Frobenius normal form of its $p$-curvature, which can be efficiently computed
  using~\cite{BoCaSc16}. Using results from the thesis of the
  author~\cite{PagPhD}, we are then able to construct an operator $L^*$
  in the same equivalence class as $L$ for which an
  $\mathrm{LCLM}$-decomposition is known. Finally, by computing an isomorphism between the
  quotient modules
  $\Fbb_q(t)\langle\partial\rangle/\Fbb_q(t)\langle\partial\rangle L^*$ and
  $\Fbb_q(t)\langle\partial\rangle/\Fbb_q(t)\langle\partial\rangle L$, we
  find a corresponding $\mathrm{LCLM}$-decomposition of $L$.
\end{abstract}



\maketitle
\section{Introduction}
Studying and solving differential equations has been an important subject
on mathematicians' mind since the invention of differential calculus and
has found many applications. Although the classical study of differential
equations concerns essentially functions of real or complex variables,
there is an algebraic counterpart to this theory, which makes sense over a
large variety of base fields, including, by not limited to, number fields,
$p$-adic fields and fields of positive characteristic. Applications include
points counting on elliptic curves~\cite{Lauder04count}, isogeny
computations~\cite{LaVa16, Eid21}, and, more generally, the study of (the
cohomology of) many arithmetic varieties. This paper focuses on linear
differential equations of the form $L(y)=0$ with
$L=f_r(t)\partial^r+f_{r-1}(t)\partial^{r-1}+\dots+f_0(t)$
and the $f_i(t)$ being rational functions over a finite field $\Fbb_q$,
where $q$ is a power of the characteristic $p$. In this setting,
differential operators of the same form of $L$ can be multiplied with one
another, but the resulting ring structure is noncommutative, as the
multiplication follows the rule $\partial t=t\partial+1$. In this context,
the question of factorisation is not only an interesting algorithmic
challenge, it also serves an important purpose, as the space of solutions
of a right factor of $L$ is a subspace of the space of solutions of $L$.
Unfortunately, unlike the commutative case of polynomials, if the operator
$L$ can be written as a product $L=L'R$, the left factor $L'$ does not
automatically give us information on the solutions of $L$. For this reason,
it is oftentimes more interesting to try to write $L$ as a common left
multiple of some operators $L_1,\dots,L_k$ where the $L_i$ are as simple as
possible, and such that the solution space of $L$ is the direct sum of the
solution spaces of the $L_i$. This is what is called an
$\mathrm{LCLM}$-decomposition of $L$, where $\mathrm{LCLM}$ stands for
\emph{least common left multiple}.\\
In characteristic $0$, linear differential equations are quite well
understood and efficient algorithms for answering all sorts of questions
about them or their solutions are available~\cite{kauers23c}. Regarding
factorisation algorithms in the case of operators with coefficients in
$\mathbb{C}(t)$, several algorithms have been proposed
(\emph{e.g.}~\cite{Gri90,VHo97,vdH07:galois,vdH07:reshol,ChGoMe22})  Those algorithms usually rely
on tools such as the monodromy group which do not have an obvious analogue
in positive characteristic. In characteristic $p$, the first significant
work towards a factorisation algorithm is the classification of finite
dimensional differential modules published by Marius van der
Put~\cite{Put95}. Theorem~\ref{thm_decompclass} and Corollary~\ref{cor_critiso} can
be seen as computational versions of this classification in the case of
differential operators. The question of factorisation in positive
characteristic was since then studied several times~\cite{Put96}, notably
by Thomas Cluzeau~\cite{Cluzeau03} who proposed a factorisation algorithm
for differential systems relying on the same principle as Theorem~\ref{thm_firstdecomp} 
applied to other randomly chosen elements of the Eigenring. Unfortunately,
unlike in characteristic $0$, elements of the Eigenring have their
characteristic polynomial with coefficients in $\Fbb_q(t^p)$ (and not
$\Fbb_q$) over which most polynomials do not split which probably explains
why this method does not work in practice. Other proposed methods
included~\cite{GiZh03} who relied on being able to find zero divisors in
the Eigenring. Unfortunately the proposed algorithm relied on hypotheses
which are not usually satisfied for operators for which
Theorem~\ref{thm_firstdecomp} does not yield a complete factorisation, as
was shown later in~\cite{GoTLoNa19}. To our knowledge, the most reliable method of
factorisation relies on a method proposed by van der Put
to solve equations of the form $\frac{\ud^{p-1}}{\ud t^{p-1}}f+f^p=a^p$, of unknown variable
$f$, in some
separable extension of $\Fbb_q(t)$. This method relied on the operator
being factored to not be of the form $N(t^p,\partial^p)$ with
$N\in\Fbb_q(X)[Y]$ and usually yielded solutions of size linear in $p$.
In their thesis~\cite{PagPhD}, the author of the paper developed an
algorithm to find solutions of this equation of size independent from $p$,
without relying on such an hypothesis. We build on this work and propose
a new method to compute $\mathrm{LCLM}$-decompositions of operators under
sufficient separability hypotheses.\\
After reestablishing some basic properties of differential operators and
$\mathrm{LCLM}$s in section~\ref{sec_prelim}, we show in section~\ref{sec_pcurv} how the $p$-curvature of an operator can
be used to fully determine the shape of an $\mathrm{LCLM}$-decomposition of
$L$ as well as its equivalence class. Section~\ref{sec_buildrepr} show how
to use the results of~\cite{PagPhD} to construct a ``small'' operator whose
decomposition is known, in the same equivalence class as $L$. Finally we
present our $\mathrm{LCLM}$-decomposition algorithm in
section~\ref{sec_algo}. The bound on the size of the output of our
algorithm is at the moment quadratic in $p$ due to having to compute a
set of differential homomorphisms between two cyclic differential modules.
Experiments have shown that this bound can probably be lowered to
be quasilinear in $p$ without modifying the algorithm itself.

\section{Preliminaries}\label{sec_prelim}
In this section, we recall the basic notions around differential operators
as well as $\mathrm{LCM}$, or rather $\mathrm{LCLM}$, decompositions. Let
$p$ be a prime number and $q=p^n$. We consider linear differential
operators with coefficients in the rational function field
$\mathbb{F}_q(t)$ where $\mathbb{F}_q$ is the finite field with $q$
elements. The ring of linear differential operators over $\mathbb{F}_q(t)$
is denoted $\mathbb{F}_q(t)\langle\partial\rangle$, and its elements are
polynomials in $\partial$ of the form
\[f_r\partial^r+f_{r-1}\partial^{r-1}+\dots+f_1\partial+f_0\]
where $f_i\in\mathbb{F}_q(t)$. The (noncommutative) multiplication in
this ring is derived from the Leibniz rule $\partial f=f\partial+\frac{\ud
}{\ud t}f$ for all $f\in\mathbb{F}_q(t)$. For any operator
$L=f_r\partial^r+\dots+f_0\in \mathbb{F}_q(t)\langle\partial\rangle\backslash\{0\}$, the
order of $L$, which we denote by $\ord(L)$, is the largest $i\in\mathbb{N}$ for
which $f_i\neq 0$. We recall that the order provides
$\mathbb{F}_q(t)\langle\partial\rangle$ with a right (and a left) Euclidean
division which makes $\mathbb{F}_q(t)\langle\partial\rangle$ a principal
left (and right) ideal domain. Two operators $L_1$ and $L_2$ generate the
same left ideal of $\mathbb{F}_q(t)\langle\partial\rangle$ if and only if there
exists $u\in\mathbb{F}_q(t)^\times$ such that $L_1=uL_2$. Furthermore this
allows us to define notion of \emph{least common left multiples}
($\mathrm{LCLM}$) and \emph{greatest common right divisors}
($\mathrm{GCRD}$).
\begin{defi}
  Let $L_1,L_2,L\in\mathbb{F}_q(t)\langle\partial\rangle$. We say that:
  \begin{enumerate}[label=\roman*)]
    \item $L=\mathrm{LCLM}(L_1,L_2)$ if and only if
      $\mathbb{F}_q(t)\langle\partial\rangle L=\mathbb{F}_q(t)\langle\partial\rangle
      L_1\cap\mathbb{F}_q(t)\langle\partial\rangle L_2$.
    \item $L=\mathrm{GCRD}(L_1,L_2)$ if and only if
      $\mathbb{F}_q(t)\langle\partial\rangle L=\mathbb{F}_q(t)\langle\partial\rangle
      L_1+\mathbb{F}_q(t)\langle\partial\rangle L_2$.
  \end{enumerate}
\end{defi}
\begin{notation}
  Let $L\in\mathbb{F}_q(t)\langle\partial\rangle$. We denote by
  $\mathcal{D}_L$ the left quotient module
  $\nicefrac{\Fbb_q(t)\langle\partial\rangle}{\Fbb_q(t)\langle\partial\rangle
  L}$.\\
  In some instances, we may want to extend the coefficient field to
  some separable extension $K$ of $\Fbb_q(t)$. When no conflict of notation
  is possible we will write $\Dcal_L$ for
  $\nicefrac{K\langle\partial\rangle}{K\langle\partial\rangle L}$ all the
  same.
\end{notation}
We state without proof the following result which describes in more details the
relations between $\mathrm{LCLM}$s and $\mathrm{GCRD}$s.
\begin{proposition}\label{short_exact_sequence}
  The following sequence of left 
  $\mathbb{F}_q(t)\langle\partial\rangle$-modules is exact:
  \[\begin{array}{rcl}
    0\rightarrow\Dcal_{\mathrm{LCLM}(L_1,L_2)}\rightarrow\Dcal_{L_1}\oplus\Dcal_{L_2}&\rightarrow&
    \Dcal_{\mathrm{GCRD}(L_1,L_2)}\rightarrow 0\\
    (M,N)&\mapsto&M-N
  \end{array}.\]
  In particular:
  \begin{enumerate}[label=\roman*)]
    \item $\ord(\mathrm{LCLM}(L_1,L_2))+\ord(\mathrm{GCRD}(L_1,L_2))=\ord
      L_1+\ord L_2$.
    \item If $\mathrm{GCRD}(L_1,L_2)=1$ then
      $\Dcal_{\mathrm{LCLM}(L_1,L_2)}\simeq \Dcal_{L_1}\oplus\Dcal_{L_2}$.
  \end{enumerate}
\end{proposition}
\begin{defi}
  We say that $L\in\Fbb_q(t)\langle\partial\rangle$ is decomposable if and
  only if there exists
  $L_1,L_2\in\Fbb_q(t)\langle\partial\rangle\backslash\Fbb_q(t)$ coprime
  (meaning that $\mathrm{GCRD}(L_1,L_2)=1$) and $L=\mathrm{LCLM}(L_1,L_2)$.
  This is equivalent to saying that $\Dcal_L$ is decomposable as a
  left $\Fbb_q(t)\langle\partial\rangle$-module (ie can be written as a
  direct sum
  of nontrivial submodules).\\
  We say that $L$ is indecomposable if it is not decomposable.
\end{defi}
\begin{remark} An irreducible operator is indecomposable but the converse
  is not true. For example the operator $\partial
  x\partial=x\partial^2+\partial$ is indecomposable but not irreducible.
\end{remark}
\begin{defi}
  Let $L\in\Fbb_q(t)\langle\partial\rangle$ and
  $(L_1,\dots,L_k)\in\Fbb_q(t)\langle\partial\rangle^k$ for some
  $k\in\mathbb{N}$. We say that $(L_1,\dots,L_k)$ is an
  $\mathrm{LCLM}$-decomposition of $L$ if and only if each $L_i$ is
  indecomposable, $L=\mathrm{LCLM}_{i=1}^k L_i$ and $\ord L=\sum_{i=1}^k
  \ord L_i$.
\end{defi}
The goal of this paper is to present an algorithm to find, given some
$L\in\Fbb_q(t)\langle\partial\rangle$, an
$\mathrm{LCLM}$-decomposition of it.
\begin{proposition}\label{prop_lclm_directsum}
  Let $L\in\Fbb_q(t)\langle\partial\rangle$ and $L_1,\dots,L_k$ be right
  divisors of $L$. $L=\mathrm{LCLM}_{i=1}^k L_i$ with $\ord
  L=\sum_{i=1}^k\ord L_i$ if and only if the canonical homomorphism
  $\Dcal_L\rightarrow \bigoplus_{i=1}^k \Dcal_{L_i}$ is an isomorphism.
\end{proposition}
\begin{proof}
  Let us first assume that $\Dcal_L\rightarrow \bigoplus_{i=1}^k
  \Dcal_{L_i}$ is an isomorphism. By definition of $L':=\mathrm{LCLM}_{i=1}^k L_i$, 
  it is a right divisor of $L$ and this
  canonical morphism factors into
  $\Dcal_L\twoheadrightarrow\Dcal_{L'}\rightarrow
  \bigoplus_{i=1}^k\Dcal_{L_i}$. Thus it can only be an isomorphism if both arrows are, and thus if
  $L=L'$. We get the equality on the orders from dimensional analysis. Let
  us now assume that $L=\mathrm{LCLM}_{i=1}^k L_i$ and $\ord L=\sum_{i=1}^k
  \ord L_i$.
  From Proposition~\ref{short_exact_sequence}, we know the result to be
  true for $k=2$. If $k\geq 3$ then we can set
  $L'=\mathrm{LCLM}_{i=1}^{k-1}L_i$. We then have
$L=\mathrm{LCLM}(L',L_k)$. From Proposition~\ref{short_exact_sequence}(ii)
  we know that $\ord L'\leq \sum_{i=1}^{k-1}\ord L_i$ and $\ord L\leq \ord
  L'+\ord L_k\leq \sum_{i=1}^k\ord L_i$. But by hypothesis, those must
  in fact be equalities so we have $\ord L'=\sum_{i=1}^{k-1}\ord L_i$. By
  induction we deduce that $\Dcal_{L'}\simeq \bigoplus_{i=1}^{k-1}
  \Dcal_{L_i}$ and $\Dcal_L\simeq\Dcal_{L'}\oplus\Dcal_{L_k}\simeq
  \bigoplus_{i=1}^{k}\Dcal_{L_i}$.
\end{proof}
\section{$p$-curvature}\label{sec_pcurv}
  In this section we recall some result used in~\cite{Put95}
  and~\cite{Cluzeau03} to find a first decomposition of a given operator
  from its $p$-curvature, and more precisely, the characteristic polynomial
  of its $p$-curvature. An efficient algorithm to compute the
  characteristic polynomial of the $p$-curvature of a given operator in
  polynomial time in the order of the operator and the degree of its
  coefficients and quasilinear time in $\sqrt{p}$ was presented
  in~\cite{BoCaSc14}. We then show in Theorem~\ref{thm_decompclass} that
  the shape of a $\mathrm{LCLM}$-decomposition of an operator can be
  entirely deduced from the Frobenius (or rational) normal
  form of its $p$-curvature (see \emph{e.g.}~\cite[section~21.4]{BJN94} for
  reminders about Frobenius normal forms of linear endomorphisms). To
  compute this normal form we use~\cite{BoCaSc16} which achieves a similar complexity. 
  \begin{defi}\cite[p.329]{PuSi03}
    Let $L\in\Fbb_q(t)\langle\partial\rangle$. The $p$-curvature of $L$ is
    the $\Fbb_q(t)$-linear endomorphism of $\Dcal_L$ defined by the
    multiplication by $\partial^p$. We denote this endomorphism by
    $\psi^L_p$.
  \end{defi}
  It is a known fact that for any $L\in\Fbb_q(t)\langle\partial\rangle$,
  there exists a $\Fbb_q(t)$-basis of $\Dcal_L$ in which the matrix of
  $\psi_p^L$ has coefficient in $\Fbb_q(t^p)$~\cite[Proposition
  2.1.2]{Dw90}. In particular, its characteristic polynomial and
  Frobenius normal form have coefficients in $\Fbb_q(t^p)$. We need to know two
  additional crucial informations about the characteristic polynomial of
  the $p$-curvature before we can state the first decomposition theorem.
  \begin{defi}
    We say that an operator $L\in\Fbb_q(t)\langle\partial\rangle$ is
    central if it commutes multiplicatively with all the elements of
    $\Fbb_q(t)\langle\partial\rangle$. Equivalently,
    $L\in\Fbb_q(t^p)[\partial^p]$~\cite[Lemma~1.1]{Put95}.
  \end{defi}
  \begin{lemma}\cite[Lemma~3.9]{BoCaSc14}\label{lm_multrednorm}
    The map that to $L\in\Fbb_q(t)\langle\partial\rangle$ associates the
    characteristic polynomial of its $p$-curvature
    $\chi(\psi^L_p)\in\Fbb_q(t^p)[Y]$ is multiplicative. Furthermore, if
    $L$ is central and monic then $\chi(\psi^L_p)(\partial^p)=L^p$.
  \end{lemma}

\begin{theorem}\label{thm_firstdecomp}
  Let $L\in\Fbb_q(t)\langle\partial\rangle$. Let
  $N_1,\dots,N_n\in\Fbb_q(t^p)[X]$ be pairwise distincts irreducible monic
  polynomials and $\nu_1,\dots,\nu_n\in\mathbb{N}$ be
  such that $\chi(\psi_p^L)=N_1^{\nu_1}\dots N_n^{\nu_n}$. We set
  $L_i:=\mathrm{GCRD}(L,N_i^{\nu_i}(\partial^p))$ for all $i\in\{1,\dots,n\}$.
  Then:
  \begin{enumerate}[label=\roman*)]
    \item $L=\mathrm{LCLM}_{i=1}^n L_i$.
    \item $\ord L=\sum_{i=1}^n\ord L_i$.
    \item For all $i\in\{ 1,\dots,n\}$,
      $\chi(\psi_p^{L_i})=N_i^{\nu_i}$.
  \end{enumerate}
\end{theorem}
\begin{proof}
  As per the kernel decomposition lemma, we can write
  $\Dcal_L=\ker N_1^{\nu_1}(\psi_p^L)\oplus\dots\oplus
  \ker N_n^{\nu_n}(\psi_p^L)$. Let $i\in\{1,\dots,n\}$. $V_i:=\ker
  N_i^{\nu_i}(\psi_p^L)$ is isomorphic to a quotient module of $\Dcal_L$
  and is thus isomorphic to some $\Dcal_{L_i'}$ (as a
  $\Fbb_q(t)\langle\partial\rangle$-module) where $L_i'$ is a right
  divisor of $L$. Furthermore since $\psi^L_{p|V_i}=\psi^{L_i'}_p$ through
  this isomorphism and by definition of $V_i$,
  $N_i^{\nu_i}(\psi^{L_i'}_p)(1)=0=N_i^{\nu_i}(\partial^p)\mod L_i'$, it
  follows that $L_i'$ is also a divisor of $N_i^{\nu_i}(\partial^p)$, thus
  it is a divisor of $L_i$. Finally we know that
  $\chi(\psi^{L_i'}_p)=\chi(\psi^L_{p|V_i})=N_i^{\nu_i}$. Let us write
  $L_i=Q_iL_i'$ for some $Q_i\in\Fbb_q(t)\langle\partial\rangle$. Since by
  definition $L_i$ is a divisor of $L$, by Lemma~\ref{lm_multrednorm},
  $\chi(\psi^{L_i}_p)|\chi(\psi^L_p)$. Furthermore
  $\chi(\psi^{L_i}_p)=\chi(\psi^{Q_i}_p)\chi(\psi^{L_i'}_p)=\chi(\psi^{Q_i}_p)N_i^{\nu_i}$.
  But since $Q_i$ is a divisor of $L_i$ it has to be a divisor of
  $N_i^{\nu_i}(\partial^p)$. Again by multiplicativity and because
  $\chi(\psi^{N_i^{\nu_i}(\partial^p)}_p)=N_i^{p\nu_i}$, it follows that
  $\chi(\psi^{Q_i}_p)$ has to be a power of $N_i$. If $Q_i\notin\Fbb_q(t)$
  then this means that this power is nonzero and $\chi(\psi^L_p)$ has to be
  divided by a higher power of $N_i$ which is not possible. Thus we can
  assume that $L_i'=L_i$. The result now follows from
  Proposition~\ref{prop_lclm_directsum}
\end{proof}
From now on, all the operators we consider will follow the following
hypothesis, even if not explicitly stated:
\begin{hypo}\label{hypo_sep}
  Let $L\in\Fbb_q(t)\langle\partial\rangle$. We assume that
  $\chi(\psi^L_p)$ has no inseparable irreducible factor over
  $\Fbb_q(t^p)$.
\end{hypo}
From Theorem~\ref{thm_firstdecomp} it follows that we can now assume that
we are working with an operator $L$ which is a divisor of some
$N(\partial^p)^{\nu}$ where $N\in\Fbb_q(t^p)[Y]$ is irreducible and
$\nu\in\mathbb{N}$. In the case where $\nu=1$ this makes $\Dcal_L$ a
left $\Dcal_{N(\partial^p)}$-module. We can show that
$\Dcal_{N(\partial^p)}$ a central simple
$C_N:=\nicefrac{\Fbb_q(t^p)[Y]}{N(Y)}$-algebra of dimension
$p^2$~\cite[Lemma~1.2]{Put95}. From
Wedderburn's theorem~\cite[Theorem~2.1.3]{GiSz06} we deduce that $\Dcal_{N(\partial^p)}$
is either a division algebra or isomorphic to $M_p(C_N)$. In the first case
it follows that $N(\partial^p)$ is irreducible, for any nontrivial divisor would be a zero divisor 
in $\Dcal_{N(\partial^p)}$. In the latter case, we know that all simple
$M_p(C_N)$-modules are isomorphic to $C_N^p$, from which follows that
divisors of $N(\partial^p)$ are irreducible if and only if they are of
order $\deg(N)$.\\
This dichotomy, which serves us well when looking for irreducible right
factors of $L$ is still very helpful as we now demonstrate.\par
\vspace{1mm}
From now on and until page 6, we assume $L$, $N$ and $\nu$ fixed. We also assume $N$ to be
separable over $\Fbb_q(t^p)$, which is to say that $L$ satisfies
Hypothesis~\ref{hypo_sep}.
\begin{notation}
  We denote $C_N\simeq \nicefrac{\Fbb_q(t^p)[Y]}{N(Y)}$ as previously stated.
  We also denote by $y_N$ the image of $Y$ in $C_N$, so that $N(y_N)=0$.\\
  We denote $K_N=\Fbb_q(t)[y_N]$. It is a differential extension of
  $\Fbb_q(t)$ since $N$ is separable over $\Fbb_q(t^p)$.
\end{notation}
\begin{remark}Note that $K_N$ is not the same as $C_N$ as it is the
  extension of $\Fbb_q(t)$, and not $\Fbb_q(t^p)$, generated by a root of
  $N$.\par
  \vspace{1mm}
  From a computational standpoint, $N$ is
  better computationally represented by a polynomial $N_*$ over $\Fbb_q(t)$ such that
  $N_*^p(Y)=N(Y^p)$, $N$ has coefficients in $\Fbb_q(t^p)$. With these
  notations $N_*(y_N^{1/p})=0$ and we use the representation
  $K_N=\Fbb_q(t)[y_N^{1/p}]\simeq\Fbb_q(t)[Y]/N_*(Y)$.
\end{remark}
\begin{lemma}
  Let $F\in K_N[\partial^p]$. Then
  $(\partial-F)^p=\partial^p-F^{(p-1)}-F^p$ where the derivative onto $F$
  is applied coefficient-wise.
\end{lemma}
\begin{proof}
  We apply~\cite[Lemma~1.4.2(1)]{Put95} to operators with coefficients in
  $K_N[T]$ provided with the coefficient-wise extension of $K_N$'s
  derivation. We then have that for any $F\in
  K_N[T]\langle\partial\rangle$,
  $(\partial-F)^p=\partial^p-\tau(F)=\partial^p-F^{(p-1)}-F^p$ using
  \cite[Lemma~1.4.2]{Put95}'s notation. We then evaluate the equality in
  $T=\partial^p$.
\end{proof}
\begin{theorem}\label{thm_isoextmorita}
  $\Dcal_{N(\partial^p)^\nu}\simeq\Dcal_{N(\partial^p)}[T]/T^\nu$.
\end{theorem}
\begin{proof}
  Note that the canonical injection
  $\Fbb_q(t)\langle\partial\rangle\hookrightarrow
  K_N\langle\partial\rangle$ induces a ring isomorphism $\Dcal_{N(\partial^p)^\nu}\xrightarrow{\sim}
  K_N\langle\partial\rangle/(\partial^p-y_N)^\nu.$ The injectivity can be
  seen by writing operators in $\Fbb_q(t)\langle\partial\rangle$ in the
  $\Fbb_q(t^p)[\partial^p]$ basis $(x^i\partial^j)_{0\leq i,j<p}$ and using
  the fact that $N$ is the minimal polynomial of $y_N$ over $\Fbb_q(t^p)$.
  The surjectivity comes from dimensional analysis. Thus by extending the
  base field we may assume with no loss of generality that $N=Y-y_N$. We construct a ring
  homomorphism from $\Dcal_{N(\partial^p)}$ in $\Dcal_{N(\partial^p)^\nu}$
  that will map $K_N$ onto itself. Such a morphism is uniquely given by the
  image of $\partial$ which is an element $A\in K_N\langle\partial\rangle$
  such that $A^p-y_N\equiv 0\mod N(\partial^p)^\nu$. We seek $A$ of the form
  $\partial+x^{p-1}F$ with $F\in C_N[\partial^p]$. Let us assume that we have
  found $F_0\in K_N[\partial^p]$ such that $(\partial+x^{p-1}F_0)^p-y_N\equiv 0\mod
  N(\partial^p)^m$ for some $m\geq 1$. Then there exists $P\in
  C_N[\partial^p]$ such that $(\partial+x^{p-1}F_0)^p-y_N=PN(\partial^p)^m$. We
  set $F_1=F_0+PN(\partial^p)^m$. Then
  \begin{align*}
    (\partial+x^{p-1}F_1)^p&-y_N=\partial^p-F_1+x^{p(p-1)}F_1-y_N\\
    =&\partial^p-F_0-PN(\partial^p)^m+x^{p(p-1)}F_0\\&+x^{p(p-1)}P^pN(\partial^p)^{pm}-y_N\\
    \equiv&(\partial+x^{p-1}F_0)^p-y_N-PN(\partial^p)^m\mod
    N(\partial^p)^{m+1}\\
    \equiv& 0\mod N(\partial^p)^{m+1}
  \end{align*}
  Starting with $F_0=0$ and $m=1$ we can construct $F$ such
  that $N(\partial+x^{p-1}F)^{p})\equiv 0\mod N(\partial)^\nu$ and thus a
  nonzero morphism $\mu:\Dcal_{N(\partial^p)}\rightarrow\Dcal_{N(\partial^p)^\nu}$. Since
  $\Dcal_{N(\partial^p)}$ is a central simple $C_N$-algebra, it has no
  nontrivial bilateral ideal thus the morphism is injective. We map $T$
  to $N(\partial^p)$ to get the morphism
  $\Dcal_{N(\partial^p)}[T]/T^\nu\rightarrow\Dcal_{N(\partial^p)^\nu}$.
  This morphism is also injective. Indeed if $A=\sum_{k=0}^{\nu-1}A_k T^k$
  with $A_K\in\Dcal_{N(\partial^p)}$ is mapped to zero, then
  $\sum_{k=0}^{\nu-1}\mu(A_k)N(\partial^p)^k\equiv 0\mod
  N(\partial^p)^\nu$. If the $A_k$ where not all zero then this means that
  there would be a lowest $k$ for which $A_k$ is not zero. But then
  $\mu(A_k)$ would have to be dividable by $N(\partial^p)$. This is not
  possible because $\mu^{-1}(\Dcal_{N(\partial^p)^m}N(\partial^p))$ is a
  bilateral ideal of $\Dcal_{N(\partial^p)}$ which does not contain $1$
  therefore is reduced to zero. Thus $A_k=0$ for all $k$ and we have an
  injective ring homomorphism. We conclude by equality of the dimensions
  over $C_N$.
\end{proof}
Before stating the structure theorem which will conclude this section and
which we will use later to compute controlled representatives of the
equivalence class of $L$, we must state without proof a result of categorical equivalence
which is a classical exemple of Morita's equivalence between rings. We
present the result in the restrictive setting where the base ring is
commutative, even though this is not actually a requirement, so that no
ambiguity may be found on the meaning of the tensor product.
Further information on Morita equivalences of rings may be found for example
in~\cite[Chapter~6]{Anderson1974}.
\begin{theorem}[Morita]\label{thm_Morita}
  Let $R$ be a commutative ring and $n\in\mathbb{N}$. The covariant functor 
  \[\begin{array}{crcl}
    \mathcal{M}\mathrm{or}:&\mathrm{Mod}_{R}&\rightarrow&\mathrm{Mod}_{M_n(R)}\\
    &M&\mapsto&R^n\otimes_{R}M
  \end{array}.\]
  is a categorical equivalence which preserves direct sums.
\end{theorem}
\begin{theorem}\label{thm_decompclass}
  Let $L\in \Fbb_q(t)\langle\partial\rangle$ be a monic divisor of some
  $N(\partial^p)^\nu$ where $N\in\Fbb_q(t^p)[Y]$ is a monic irreducible
  polynomial, separable over $\Fbb_q(t^p)$ and $\nu\in\mathbb{N}$.
  \begin{itemize}
    \item If $N(\partial^p)$ is irreducible in
      $\Fbb_q(t)\langle\partial\rangle$ then $L=N(\partial^p)^m$ with
      $m\leq \nu$ and is indecomposable.
    \item Otherwise let $N^{\nu_1}|N^{\nu_2}|\dots|N^{\nu_k}$ be the Frobenius
      invariants of $\psi^L_p$. Then there exists $L_1,\dots,L_k\in
      \Fbb_q(t)\langle\partial\rangle$ such that
      \begin{enumerate}[label=\roman*)]
        \item $L=\mathrm{LCLM}_{i=1}^k L_i$ and $\ord L=\sum_{i=1}^k \ord
          L_i$.
        \item $\ord L_i=\nu_i\deg N$ for all $i\in\{1,\dots,k\}$.
        \item Each $L_i$ is indecomposable
      \end{enumerate}
  \end{itemize}
\end{theorem}
\begin{proof}
      Let us first assume that $N(\partial^p)$ is irreducible. We can
      divide $L$ by its highest power of $N(\partial^p)$ which divides it,
      so we may as well assume that $L$ and $N(\partial^p)$ are coprime and
      show that $L\in\Fbb_q(t)$.
      Then the multiplication by $N(\partial^p)$, which is $N(\psi^L_p)$ induces an automorphism of
      $\Dcal_L$. However, since $N(\partial^p)^\nu$ is a multiple of $L$,
      $N(\psi^L_p)$ is nilpotent which is only possible if $\Dcal_L=\{0\}$
      which is to say that $L\in\Fbb_q(t)$. Since any two nontrivial right factors of $L$
      would also be powers of $N(\partial^p)$, their $\mathrm{GCRD}$ would
      never be $1$ which proves that $L$ is indecomposable.\par
      We now assume that $N(\partial^p)$ is reducible which is to say that
      $\Dcal_{N(\partial^p)}\simeq M_p(C_N)$. Thus we according to
      Theorem~\ref{thm_isoextmorita}, $\Dcal_{N(\partial^p)^\nu}\simeq
      \Dcal_{N(\partial^p)}[T]/T^\nu\simeq M_p(C_N[T]/T^\nu)$. Since
      $\Dcal_L$ is a left $\Dcal_{N(\partial^p)^\nu}$-module, it
      corresponds through Morita's equivalence to a finite dimensional (over
      $C_N$) $C_N[T]/T^\nu$-module, which is to say a finite dimensional
      $C_N$-vector space $V$ provided with a nilpotent endomorphism. It
      follows that there exists $k'\in\mathbb{N}$ and $m_1\leq\dots\leq m_{k'}$
      such that $V\simeq C_N[T]/T^{m_1}\oplus\dots\oplus C_N[T]/T^{m_k'}$
      and according to Theorem~\ref{thm_Morita},
      $\Dcal_L\simeq R^p\otimes_R C_N[T]/T^{m_1}\oplus\dots\oplus R^p\otimes_R
      C_N[T]/T^{m_k'}$ as an $M_p(R)$-module, where $R=C_N[T]/T^\nu$.\\
      We set $M_i$ the submodule of $\Dcal_L$ corresponding to
      $R^p\otimes_R C_N[T]/T^{m_i}$ in that decomposition. It is
      canonically isomorphic to some $\Dcal_{L_i}$ where $L_i$ is a right
      divisor of $L$. From Proposition~\ref{prop_lclm_directsum} we have
      $L=\mathrm{LCLM}_{i=1}^{k'}L_i$ and $\ord L=\sum_{i=1}^{k'}\ord L_i$. Since the functor $\mathcal{M}\mathrm{or}$ from
      Theorem~\ref{thm_Morita} preserves direct sums and $C_N[T]/T^{m_i}$
      is an indecomposable $C_N[T]$-module, $L_i$ is indecomposable.
      Furthermore, the map $N(\psi^L_{p|M_i})$ corresponds to the
      multiplication by $T$ on $R^p\otimes_R C_N[T]/T^{m_i}$, thus the
      minimal polynomial of $\psi^L_{p|M_i}$ is $N^{m_i}$. Since
      \begin{align*}
        \dim_{\Fbb_q(t)}M_i&=p^{-1}\dim_{\Fbb_q(t^p)}M_i\\
        &=p^{-1}\deg(N)\dim_{C_N}R^p\otimes_R
      C_N[T]/T^{m_i}\\
        &=\deg(N)p^{-1}pm_i=m_i\deg(N),
      \end{align*} it follows that
      $\psi^{L}_{p|M_i}$ is cyclic. Thus $N^{m_1}|\dots|N^{m_{k'}}$ fit the
      criteria of the Frobenius invariants of $\psi^L_p$ and by unicity we
      have $k'=k$ and $m_i=\nu_i$ for all $i$.
\end{proof}

\section{Equivalence class and representative}\label{sec_buildrepr}
Theorem~\ref{thm_decompclass} allows us to know precisely, given
$L\in\Fbb_q(t)\langle\partial\rangle$, what one of its
$\mathrm{LCLM}$-decompositions should look like, meaning that we know
precisely the degrees of its coefficients. It is actually a well-known fact
that two $\mathrm{LCLM}$-decompositions of the same operator must
have factors of the same degrees. In positive characteristic, this can
actually be seen as a Corollary of Theorem~\ref{thm_decompclass}, since an
$\mathrm{LCLM}$-decomposition of $L$ also gives the Frobenius normal form
of $\psi^L_p$. In fact, as we now show, the equivalence class of $\psi^L_p$
entirely determines the isomorphism class of $\Dcal_L$ as an
$\Fbb_q(t)\langle\partial\rangle$-module.
\begin{defi}
  Let $L_1,L_2\in\Fbb_q(t)\langle\partial\rangle$. We say that $L_1$ and
  $L_2$ are equivalent if and only if $\Dcal_{L_1}$ and $\Dcal_{L_2}$ are
  isomorphic as $\Fbb_q(t)\langle\partial\rangle$-modules.
\end{defi}
\begin{lemma}\label{lem_isoindec}
  If $L_1$ and $L_2$ are two indecomposable operators in
  $\Fbb_q(t)\langle\partial\rangle$ satisfying Hypothesis~\ref{hypo_sep}, then they are equivalent if
  and only if $\chi(\psi^{L_1}_p)=\chi(\psi^{L_2}_p)$.
\end{lemma}
\begin{proof}
  The fact that $\Dcal_{L_1}\simeq\Dcal_{L_2}$ implies
  $\chi(\psi^{L_1}_p)=\chi(\psi^{L_2}_p)$ is immediate. Conversely, let us
  assume that $\chi(\psi^{L_1}_p)=\chi(\psi^{L_2}_p)$. This implies in
  particular that $\ord(L_1)=\ord(L_2)$. Since $L_1$ and
  $L_2$ are indecomposable, according to Theorem~\ref{thm_firstdecomp},
  $\chi(\psi^{L^*}_p)$ is of the form $N^\nu$ where $N$ is an irreducible
  polynomial over $\Fbb_q(t^p)$. Then, if $N(\partial^p)$ is irreducible we
  know from Theorem~\ref{thm_decompclass} that we must have
  $L_i=N(\partial^p)^{m_i}$ for some $m_i$ and $i\in\{1,2\}$. By degree
  equality we have $m_1=m_2$ and $L_1=L_2$. If now $N(\partial^p)$ is
  reducible then $\Dcal_{N(\partial^p)}\simeq M_p(C_N)$ and
  $\Dcal_{N(\partial^p)^\nu}\simeq M_p(C_N[T]/T^\nu)$ according to
  Theorem~\ref{thm_isoextmorita}. Thus $\Dcal_{L_1}$ and $\Dcal_{L_2}$
  correspond through Morita's equivalence to two indecomposable (see
  Theorem~\ref{thm_Morita}, $\mathcal{M}\mathrm{or}$ preserves direct sums)
  $C_N[T]/T^\nu$-modules, $C_N[T]/T^{\nu_1}$ and $C_N[T]/T^{\nu_2}$
  respectively. It follows that $\Dcal_{L_i}\simeq R^p\otimes_R
  C_N[T]/T^{\nu_i}$ for $i\in\{1,2\}$ and $R=C_N[T]/T^\nu$. By equality of
  the orders and thus of the dimensions, we must have $\nu_1=\nu_2$.
\end{proof}
\begin{corollary}\label{cor_critiso}
  Two operators $L_1,L_2\in\Fbb_q(t)\langle\partial\rangle$ (not necessarily divisor
  of $N(\partial^p)^\nu$) satisfying Hypothesis~\ref{hypo_sep} are equivalent if and only if $\psi^{L_1}_p$ and
  $\psi^{L_2}_p$ are equivalent as $\Fbb_q(t)$-linear maps.
\end{corollary}
\begin{proof}
  It is immediate that if $\Dcal_{L_1}$ and $\Dcal_{L_2}$ are isomorphic as
  $\Fbb_q(t)\langle\partial\rangle$-modules, then $\psi^{L_1}_p$ and
  $\psi^{L_2}_p$ are equivalent. Conversely, we now assume $\psi^{L_1}_p$
  and $\psi^{L_2} _p$ to be equivalent. In particular they have the same
  characteristic polynomial and $L_1$ and $L_2$ have the same order. 
  We can apply Theorem~\ref{thm_firstdecomp} and reduce the problem to the case where 
  both $L_1$ and $L_2$ are divisors of $N(\partial^p)^\nu$. Then according
  to Theorem~\ref{thm_decompclass}, either $L_1$ and $L_2$ are equal to
  the same power of $N(\partial^p)$ (since they are of the same order), or
  the orders of the factors of their $\mathrm{LCLM}$-decompositions are
  given by the Frobenius invariants of their $p$-curvatures. Since
  $\psi^{L_1}_p$ and $\psi^{L_2}_p$ are equivalents, those invariants are
  the same and there exists, $k\in\mathbb{N}$ and
  $L_{i,j}\in\Fbb_q(t)\langle\partial\rangle$ indecomposable for
  $(i,j)\in\{1,2\}\times\{1,\dots,k\}$ such that
  $\Dcal_{L_i}=\Dcal_{L_{i,1}}\oplus\dots\oplus\Dcal_{L_{i,k}}$ with in
  addition $\ord L_{1,j}=\ord L_{2,j}$ for all $j$. According to
  Lemma~\ref{lem_isoindec}, $\Dcal_{L_{1,j}}\simeq\Dcal_{L_{2,j}}$ for all
  $j$ and thus $\Dcal_{L_1}\simeq\Dcal_{L_2}$.
\end{proof}
Since we now know how to compute the isomorphism class of $\Dcal_L$ from
the Frobenius normal form of its $p$-curvature (for which computation,
efficient algorithms exist, see for example~\cite{BoCaSc16}), our goal is
now to show how to compute representatives of arbitrarily chosen
isomorphism classes, of which an
$\mathrm{LCLM}$-decomposition is known. Since
once it is established whether or not $N(\partial^p)$ is reducible, the
irreducible case is trivial, we will assume for the rest of this section that
$N(\partial^p)$ is reducible. Note that in their thesis, the author of this
paper
presents a polynomial time method to determine whether or not
$N(\partial^p)$ is reducible (see~\cite[Section~3.3.3]{PagPhD}
or~\cite[Theorem~3.9]{Pag24}).
\begin{lemma}\label{lm_exindec}
  Let $n\in\mathbb{N}$. The operator $(x\partial)^{n}$ is
  indecomposable in $K\langle\partial\rangle$ where $K$ is any separable
  extension of $\Fbb_q(t)$.
\end{lemma}
\begin{proof}
  From the multiplicativity of $L\mapsto \chi(\psi^L_p)$ (see
  Lemma~\ref{lm_multrednorm}) and since $\chi(\psi^\partial_p)(Y)=Y$, it
  follows that $\chi(\psi^{(x\partial)^n}_p)(Y)=Y^n$, thus
  $\psi^{(x\partial)^n}_p$ is nilpotent. According to
  Theorem~\ref{thm_decompclass}, it is enough to show that
  $\chi_{min}(\psi^{(x\partial)^n}_p)=Y^p$ which is to say that
  $\psi^{(x\partial)^n}_p$ is nilpotent of maximal rank. This is equivalent
  to saying that $\ker \psi^{(x\partial)^n}_p=1$. Furthermore, from
  \cite[Lemma~13.2]{PuSi03}, we know that $\dim_K\ker
  \psi^{(x\partial)^n}_p=\dim_C \{f\in K|(x\partial)^n(f)=0\}$, where $C$
  is the constant field of $K$. Since $K$ is a separable extension of
  $\Fbb_q(t)$, $C$ is the set of all $p$-th powers of elements of $K$ and
  $1,x,\dots,x^{p-1}$ is a $C$-basis of $K$. Thus if $f\in K=\sum_{i=0}^{p-1}
  f_i x^i$, with $f_i\in C$, then we have
  $(x\partial)^n(\sum_{i=0}^{p-1}f_i x^i)=\sum_{i=1}^{p-1}i^nf_i x^i$ which
  can only be zero if $f_i=0$ for all $i>0$ which is to say that $f\in C$.
\end{proof}
In their thesis~\cite[Section~3.4]{PagPhD}, the author showed, given $N\in\Fbb_q(t^p)[Y]$, how
to compute a ``small'' $f_N\in K_N$ such that $(\partial-f_N)$ is a divisor
of $\partial^p-y_N$ when $N(\partial^p)$ is reducible. Since we assumed
that to be the case we fix one such $f_N$ for
Proposition~\ref{prop_buildrepr}.
\begin{prop}\label{prop_buildrepr}
  Let $(m_1,\dots,m_p)\in(\mathbb{N}\cup\{0\})^p$. For each
  $i\in\{1,\dots,p\}$ let $L_i\in\Fbb(t)_q\langle\partial\rangle$ be a left
  multiple of $(x\partial-xf_N+i)^{m_i}$ of order $m_i\deg(N)$. Let
  $L=\mathrm{LCLM}_{i=1}^p L_i$. Then
  \begin{enumerate}[label=\roman*)]
    \item $L_i$ is indecomposable.
    \item $\chi(\psi^{L_i}_p)=N^{m_i}$.
    \item $\ord L=\sum_{i=1}^p\ord L_i$.
  \end{enumerate}
  which is to say that the $L_i$ give an $\mathrm{LCLM}$-decomposition of
  $L$.
\end{prop}
\begin{proof}
  We claim that $L'_i:=(x\partial-xf_N+i)^{m_i}$ is indecomposable as the image
  of $(x\partial)^{m_i}$ by the ring automorphism of
  $K_N\langle\partial\rangle$ which is the identity on $K_N$ and maps
  $\partial$ onto $\partial-f_N+\frac{i}{x}$. Furthermore,
  $(\partial-f_N+\frac{i}{x})^p=\partial^p-(f_N+\frac{i}{x})^{(p-1)}-(f_N+\frac{i}{x})^p=\partial^p-y_N$
  (see~\cite[Lemma~1.4.2]{Put95}). Since it is indecomposable,
  $\chi_{min}(\psi^{L'_i}_p)(Y)=(Y-y_N)^{m_i}$. We see that
  $(Y-y_N)^{m_i}|\chi_{min}(\psi^{L_i}_p)$, since $L_i$ is a multiple
  of $L_i'$. Since $L_i\in\Fbb_q(t)\langle\partial\rangle$ and $N$ is the minimal
  polynomial of $y_N$ over $\Fbb_q(t^p)$, it follows that
  $N^{m_i}|\chi_{min}(\psi^{L_i}_p)$. Since $L_i$ is exactly of order
  $m_i\deg(N)$, it follows that $\psi_p^{L_i}$ is cyclic, which means that
  $L_i$ is indecomposable by Theorem~\ref{thm_decompclass}.\\
  There remains to show that $\ord L=\sum_{i=1}^p \ord L_i$. We proceed by
  induction. Let $L'=\mathrm{LCLM}_{i=1}^{p-1} L_i$ and let us assume that
  we have shown that $\ord L'=\sum_{i=1}^{p-1}\ord L_i$. We want to show
  that $\ord L=\ord L'+\ord L_p$. According to
  Proposition~\ref{prop_lclm_directsum}, it is enough to show that $L'$ and
  $L_p$ are coprime. Let $K_S$ be the splitting field of $N$ over
  $\Fbb_q(t)$. It is a Galois extension of $K_N$. If $L'$ and $L_p$ were
  not coprime, we could take $L^*\in
  K_S\langle\partial\rangle\backslash K_S$ an irreducible divisor of $L_p$
  and $L'$. Since
  $N(\partial^p)=\prod_{\sigma\in\mathrm{Gal}(K_S/\Fbb_q(t))}\partial^p-\sigma(y_N)$
  and $L^*|N(\partial^p)$, it follows from Theorem~\ref{thm_firstdecomp}
  that $L^*$ must divide $\partial^p-\sigma(y_N)$ for some
  $\sigma\in\mathrm{Gal}(K_S/\Fbb_q(t))$. With no loss of generality, we
  may assume $L^*$ to be a divisor of $\partial^p-y_N$. Then $L^*$
  must be a factor of
  $\mathrm{GCRD}(L_p,\partial^p-y_N)=\partial-f_N$ thus $L^*=\partial-f_N$.
  We claim that $\mathrm{GCRD}(\partial^p-y_N,L')=\mathrm{LCLM}_{i=1}^{p-1}
  \partial-f_N+\frac{i}{x}$. Indeed, by induction hypothesis,
  $\Dcal_{L'}\simeq \bigoplus_{i=1}^L \Dcal_{L_i}$ thus
  $\Dcal_{L'}(\partial^p-y_N)=\bigoplus_{i=1}^{p-1}\Dcal_{L_i}\partial^p-y_N$.
  It follows, by definition of the $\mathrm{GCRD}$ that
  $\mathrm{GCRD}(L',\partial^p-y_N)=\mathrm{LCLM}_{i=1}^{p-1}
  \mathrm{GCRD}(\partial^p-y_N,L_i)=\mathrm{LCLM}_{i=1}^{p-1}\partial-f_N+\frac{i}{x}$.
  We make a shift
  by $f_N$ so that we may assume that $f_N=0$ and $L^*=\partial$ and show that it is not a
  divisor of $\mathrm{LCLM}_{i=1}^{p-1}\partial+\frac{i}{x}$. But that is
  easy to see since the space of solutions in $K_S$ of
  $\mathrm{LCLM}_{i=1}^{p-1}\partial-\frac{i}{x}$ is
  $\bigoplus_{i=1}^{p-1}C_S x^i$ which does not contain $C_S$ the space of
  solutions of $\partial$ (where $C_S$ is the constant field of $K_S$).
\end{proof}
\begin{remark}
  Lemma~\ref{lm_exindec} and Proposition~\ref{prop_buildrepr} stay true when replacing $x$ by any non constant
  $g\in K$ or $\Fbb_q(t)$. The proof is the same but using the basis
  $1,g,\dots,g^{p-1}$ of $K$ instead.
\end{remark}
We no longer assume $L, N$ and $\nu$ to be those fixed in page 3.\par
\vspace{1mm}
Proposition~\ref{prop_buildrepr} can now be used, to devise an algorithm
computing, given a sequence of polynomials $P_1|\dots|P_m\in\Fbb_q(t^p)[Y]$
(which uniquely determine an equivalence class of operators according
to Corollary~\ref{cor_critiso}), a representative of the equivalence class
of operators for which the $P_i$ are the Frobenius invariants of their
$p$-curvatures. We restrict ourselves to polynomials
satisfying Hypothesis~\ref{hypo_reducible}, as other cases will be better
handled separately in the final $\mathrm{LCLM}$-decomposition algorithm
(Algorithm~\ref{algo_LCLMdec}).
\begin{hypo}\label{hypo_reducible}
  Let $Q\in\Fbb_q(t)[Y]$. We assume that all irreducible factors $N$ of $Q$
  are separable over $\Fbb_q(t)$ and such that $N^p(\partial)$ is reducible
  in $\Fbb_q(t)\langle\partial\rangle$.
\end{hypo}
\begin{algo}
  \begin{flushleft}
    \emph{Input:} $Q_1|Q_2\dots|Q_m\in\Fbb_q(t)[Y]$ with $m\leq p$ such
    that $Q_m$ respects Hypothesis~\ref{hypo_reducible}.\\
    \emph{Output:} \begin{itemize}
      \item$L^*\in\Fbb_q(t)\langle\partial\rangle$ such that the
    Frobenius invariants of $\psi^{L_*}_p$ are $P_1|P_2\dots |P_m$ with
    $P_i(X^p)=Q_i^p(X)$ for all $i$.
  \item $L^*_1,\dots,L_k^*$ a $\mathrm{LCLM}$-decomposition of $L_*$.
    \end{itemize}
  \end{flushleft}
  \BlankLine
  \begin{enumerate}
    \item Let \emph{list\_factor} be the list of irreducible factors of
      $Q_m$.
    \item \textbf{For} $N$ \textbf{in} \emph{list\_factor}:
      \begin{enumerate}
        \item Let $a_N$ be a root of $N$ in a separable closure of $\Fbb_q(t)$.
        \item Compute $f_N\in \Fbb_q(t)[a_N]$ such that
          $f_N^{(p-1)}+f_N^p=a_N^p$
          using~\cite[Algorithm~9]{PagPhD}.
      \end{enumerate}
    \item Let \emph{LCLM\_decomp} be an empty list
    \item \textbf{For} $i$ \textbf{in} $\{1,\dots,m\}$:
      \begin{enumerate}
    \item \textbf{For} $N$ \textbf{in} \emph{list\_factor}:
      \begin{enumerate}
        \item Compute $\nu_N(Q_i)$ the multiplicity of $Q_i$ in $N$.
        \item Compute $N_{i,N}$ the minimal multiple of
          $(x\partial-xf_N)^{\nu_N(Q_i)}$ in
          $\Fbb_q(t)\langle\partial\rangle$.
        \item Add $N_{i,N}(\partial+\frac{i}{x})$ to \emph{LCLM\_decomp}.
      \end{enumerate}
      \end{enumerate}
    \item Compute and return the $\mathrm{LCLM}$ of the elements of
      \emph{LCLM\_decomp} using~\cite{BCSZ12} together with \emph{LCLM\_decomp}.
  \end{enumerate}
  \caption{nice\_repr}
  \label{algo_nicerepr}
\end{algo}
\begin{defi}\label{def_degree}
  Let $L\in\Fbb_q(t)\langle\partial\rangle$. We say that that $L$ is of
  degree (smaller than) $d$ if there exist polynomials
  $D,p_0,\dots,p_{\ord(L)}\in\Fbb_q[t]$ of degree smaller than $d$ such
  that
  $L=D^{-1}(p_{\ord(L)}(t)\partial^{\ord(L)}+\dots+p_1(t)\partial+p_0(t))$.
\end{defi}
\begin{theorem}\label{thm_correctalgo1}
  Algorithm~\ref{algo_nicerepr} is correct.
  Furthermore, if $Q_m$ is a polynomial in $\Fbb_q[t,Y]$ of bidegree
  $d_t,d_Y$, then Algorithm~\ref{algo_nicerepr} terminates in time
  polynomial in $d_t,d_Y$ and $m$, and quasilinear in $p$. It outputs operators $L_*$ of degree
  polynomial in $d_t,d_Y$ and $m$., and the operators $L_1^*,\dots,L_k^*$ in
  the $\mathrm{LCLM}$-decomposition of $L_*$ it outputs have degree
  polynomial in $d_t$ and $d_Y$.\\ 
\end{theorem}
\begin{proof}
  The correctness of Algorithm~\ref{algo_nicerepr} is a direct consequence of
  Theorem~\ref{thm_firstdecomp} and Proposition~\ref{prop_buildrepr}.
  Hypothesis~\ref{hypo_reducible} ensures that
  Proposition~\ref{prop_buildrepr} applies for each $N$ in
  \emph{list\_factor}.\\
  
  Each $f_N$ in step~(3a) can be represented by as
  $$f_N=D_{f_N}^{-1}(f_{N,\deg_y(N)-1}a_N^{\deg_y(N)-1}+\dots+f_{N,0})$$ for
  some polynomials $D_{f_N},f_{N,0},\dots,f_{N,\deg_y(N)-1}$ of degrees
  less than some $d'$,
  polynomial in 
  $\deg_t(N)\leq d_t,\deg_Y(N)\leq d_Y$ (see~\cite[Theorem~3.4.29]{PagPhD}).
  Then for every integer $n\in\mathbb{N}$ $(\partial-f_N)^n$ has coefficients in
  $\Fbb_q(t)[a_N]$ which can also be represented by $\deg_Y(N)$ polynomials
  in $\Fbb_q[t]$ and a common denominator of degree smaller than a polynomial in
  $\deg_t(N),\deg_Y(N)$ and
  $n$ (in fact, since usually $\deg_t(N)\deg_Y(N)=O(d')$, this polynomial is usually $O(nd')$).\\
  A multiple of $(\partial-f_N)^n$ in
  $\Fbb_q(t)\langle\partial\rangle$ is a generator of the kernel of the
  natural
  $\Fbb_q(t)$-linear map $$\pi_N:\Fbb_q(t)\langle\partial\rangle_{\leq
  n\deg_{Y}(N)}\rightarrow
  \nicefrac{K_N\langle\partial\rangle}{K_N\langle\partial\rangle(\partial-f_N)^n}.$$
  The matrix $M$ of $\pi_N$ in the $\Fbb_q(t)$-bases
  $(1,\partial,\dots,\partial^{n\deg_Y(N)})$ and
  $(1,a,\dots,a^{\deg_Y(N)-1},\partial,a\partial,\dots,a^{\deg_Y(N)-1}\partial^{n-1})$
  is given by $$[v,\theta(v),\dots,\theta^{n\deg_Y(N)}(v)]$$ where
  $v=\,^t(1,0,\dots,0)$ and
  $\theta=\frac{\mathrm{d}}{\mathrm{d}t}+T$ with $T$ being the matrix $(A_N'D_a)\otimes I_{n}+C((\partial-f_N)^n)$
  where $A_N'$ is the matrix of the multiplication by $a_N'$ in
  $\Fbb_q(t)[a]$, $D_a$ is the matrix of $\frac{\ud}{\ud a}$ and
  \[C((\partial-f_N)^n)=\begin{pmatrix}
    &&&-F_0\\
    I_{\deg_Y(N)}&&&-F_1\\
    &\ddots&&\vdots\\
    &&I_{\deg_Y(N)}&-F_{n-1}
  \end{pmatrix}\]
  where $(\partial-f)^n=\partial^n+f_{n-1}\partial^{n-1}+\dots+f_0$ and
  $F_i$ is the matrix of the multiplication by $f_i$ in $\Fbb_q(t)[a]$.
  Each entry of $T$ has size polynomial in $n,\deg_t(N)$ and $\deg_Y(N)$ therefore
  it is also the case of an element in the kernel of $M$.\\
  Since $n\deg_Y(N)\leq d_Y$, it follows that each $L_i^*$ is of degree
  polynomial in $d_t$ and $d_Y$. The fact
  that $L^*$ has coefficients of degree polynomial in $d_t,d_Y$ and $m$ can be seen as a consequence
  of~\cite{BCSZ12} and of the fact that $\ord L^*=\sum_{i=1}^k\ord
  L_i^*\leq md_Y$
\end{proof}
\begin{remark}
  We see in the proof of the theorem that the problem of finding a multiple
  of $(\partial-f_N)^n$ in $\Fbb_q(t)\langle\partial\rangle$ is in fact a special instance
  of~\cite[Problem~1]{Gai25} although here in positive characteristic.
  Tight bounds could probably be derived from the same type of analysis.
\end{remark}

\section{Isomorphism and $\mathrm{LCLM}$-decomposition propagation}\label{sec_algo}
In this section we seek to find an isomorphism between the quotient modules
of two equivalent operators. We recall that a
morphism of $\Fbb_q(t)\langle\partial\rangle$-modules
$\varphi:\Dcal_{L_1}\rightarrow\Dcal_{L_2}$, with
$L_1,L_2\in\Fbb_q(t)\langle\partial\rangle$, is entirely determined by
$\varphi(1)$ which must verify, for $\varphi$ to be well-defined, that
there exists $Q\in\Fbb_q(t)\langle\partial\rangle$ such that $L_1
\varphi(1)=QL_2$. This makes finding a morphism between $\Dcal_{L_1}$ and
$\Dcal_{L_2}$ a textbook example of finding rational solutions of a mixed
differential equation studied by Mark van Hoeij in~\cite{Hoeij96}.
Unfortunately, their work is limited to the case of characteristic $0$. The
usual method of solving this problem in characteristic $0$ is to transform
it as finding rational solutions of a certain linear differential system
$Y'=AY$ with $A$ a matrix with rational coefficients. We can then bound the
poles of rational solutions and their multiplicities according to the
coefficients of $A$. Unfortunately, this is not possible in characteristic
$p$ as the field of constants is $\Fbb_q(t^p)$ which allow the solutions to
have arbitrarily many poles of arbitrarily high valuation. It is possible
that assuming $p$ to be very high compared to the parameters of $A$ could
allow us to use the same methods as in characteristic $0$ and seize
``small'' isomorphisms between $\Dcal_{L_1}$ and $\Dcal_{L_2}$ (that is to
say given by an operator with coefficients of degree independent from $p$).
However, $L_1$ and $L_2$ being equivalent does not guarantee the existence
of a ``small'' isomorphism. Furthermore, while the existence of a ``small''
isomorphism would, together with the work of the previous sections,
guarantee the existence of ``small'' $\mathrm{LCLM}$-decomposition (for a
similar notion of ``smallness''), the converse is not true.
\begin{example}
  If $p\geq 3$, the operators $\partial$ and $\partial+\frac{p-1}{2t}$ are
  irreducible and equivalent and the degree of their coefficients are
  independent from $p$. However, an isomorphism
  $\varphi:\Dcal_{\partial}\rightarrow\Dcal_{\partial+(p-1)/2t}$ is given
  by $\varphi(1)=g(t^p) t^{(p-1)/2}$ for any $g\in\Fbb_q(t)$. Regardless of
  the choice of $g$, the degree of the coefficients of the isomorphism is
  at least linear in $p$.
\end{example}
Instead, we solve the mixed differential equations by considering it as a
$\Fbb_q(t^p)$-linear system. This induces a quadratic dependency on $p$ 
on the ``size'' of the $\mathrm{LCLM}$-decompositions we will obtain,
though experiments suggest that this dependency is in fact only linear in
$p$.
\begin{lemma}\label{lm_decpropag}
  Let $L,L^*$ be equivalent operators in $\Fbb_q(t)\langle\partial\rangle$,
  and $L^*_1,\dots,L^*_k$ be an $\mathrm{LCLM}$ decomposition of $L^*$. For any
  isomorphism
  $\varphi:\Dcal_{L^*}\xrightarrow{\sim}\Dcal_L$,
  $L=\mathrm{LCLM}_{i=1}^{k}\mathrm{GCRD}(L,\varphi(L^*_i))$ is an
  $\mathrm{LCLM}$ decomposition of $L$.
\end{lemma}
\begin{proof}
  Since $\varphi$ is an isomorphism of
  $\Fbb_q(t)\langle\partial\rangle$-modules, it is in particular an
  isomorphism of $\Fbb_q(t)$-vector spaces. In particular it preserves the
  codimension of submodules. It maps the
  submodule $\Dcal_{L^*}L_i^*$ of $\Dcal_{L^*}$ generated by $L_i^*$ to
  the submodule of $\Dcal_L$ generated by $\varphi(L_i^*)$ which by
  definition of the $\mathrm{GCRD}$ is
  $\Dcal_{L}\mathrm{GCRD}(L,\varphi(L_i^*))$. Furthermore we have
  $\varphi(\bigcap_{i=1}^k\Dcal_{L^*}L_i^*)=\bigcap_{i=1}^k\Dcal_{L}\mathrm{GCRD}(L,\varphi(L_i^*))=\{0\}$.
  Thus,
  $L=\mathrm{LCLM}_{i=1}^k\mathrm{GCRD}(L,\varphi_i(L_i^*))$, by definition
  of the $\mathrm{LCLM}$.
  Finally we have that 
  \begin{align*}
    &\sum_{i=1}^k\ord\mathrm{GCRD}(L,\varphi(L_i^*))=\sum_{i=1}^{k}\dim_{\Fbb_q(t)}\Dcal_{\mathrm{GCRD}(L,\varphi(L_i^*))}\\
    =&\sum_{i=1}^k\mathrm{codim}_{\Fbb_q(t)}\Dcal_{L}\mathrm{GCRD}(L,\varphi(L_i^*))
    =\sum_{i=1}^k\mathrm{codim}_{\Fbb_q(t)}\Dcal_{L^*}L_i^*\\
    =&\sum_{i=1}^k\ord L_i^*
    =\ord L^*=\ord L
  \end{align*}
  Thus we do have a $\mathrm{LCLM}$ decomposition of $L$.
\end{proof}
\begin{algo}
  \begin{flushleft}
    \emph{Input:} $L\in\Fbb_q(t)\langle\partial\rangle$ satisfying
    Hypothesis~\ref{hypo_sep}.\\
    \emph{Output:} $L_1,\dots,L_k$ a
    $\mathrm{LCLM}$-decomposition of $L$.
  \end{flushleft}
  \BlankLine
  \begin{enumerate}
    \item Compute $Q_1|Q_2|\dots|Q_m$ such that if
      $P_1|\dots|P_m=\chi_{min}(\psi^L_p)$ the Frobenius invariants of
      $\psi^L_p$, then $Q_i^p(Y)=P_i(Y^p)$ using~\cite{BoCaSc16}
    \item Let \emph{list\_factor} be the list of irreducible factors of
      $Q_m$.
    \item Let \emph{LCLM\_decomp} be an empty list
    \item \textbf{If} $m==p$ \textbf{do}:
      \begin{itemize}
    \item \textbf{For} $N_*$ \textbf{in} \emph{list\_factor}:
      \begin{enumerate}
        \item \textbf{Let} $\nu_{N_*}$ be the valuation associated to $N_*$
          in $\Fbb_q(t)[Y]$.
        \item \textbf{If} $\nu_{N_*}(Q_1)==\nu_{N_*}(Q_m)$ \textbf{do:}
              \begin{enumerate}
                \item \textbf{Let} $\nu=\nu_{N_*}(Q_m)$
                \item \textbf{Replace} $L$ with $L\cdot
                  N_*^{-p\nu}(\partial)$
                \item \textbf{Replace} $Q_i$ with $Q_i\cdot N_*^{-\nu}$ for all $i$.
                \item \textbf{If} $N_*^p(\partial)$ is irreducible
              (\cite[Algorithm~6]{PagPhD}) \textbf{add} $N_*^{p\nu}(\partial)$ to
                  \emph{LCLM\_decomp}\\
                 \textbf{Else}:
                  \begin{enumerate}
                    \item Use Algorithm~\ref{algo_nicerepr} on
                      $\underbrace{N_*^{\nu}|\dots|N_*^\nu}_{p\text{
                        times}}$. Let $L^*$ be the result,
                      $L_1^*,\dots,L_p^*$ be its
                      $\mathrm{LCLM}$-decomposition.
                    \item \textbf{Add} $L_1^*,\dots,L_p^*$ to
                      \emph{LCLM\_decomp}
                  \end{enumerate}
                \item \textbf{Remove} $N_*$ from \emph{list\_factor} 
              \end{enumerate}
      \end{enumerate}
      \end{itemize}
    \item Use Algorithm~\ref{algo_nicerepr} on $Q_1|\dots|Q_m$. Let $L^*$
      be the output and $L_1^*,\dots,L_k^*$ be its
      $\mathrm{LCLM}$-decomposition.
    \item Compute the kernel of the $\Fbb_q(t^p)$-linear map
      \[\begin{array}{crcl}
        \Lambda_{L^*,L}:&\Fbb_q(t)\langle\partial\rangle_{< \ord
        L}&\rightarrow&\Dcal_L\\
        &M&\mapsto&L^*M\mod L
      \end{array}\]
    \item \textbf{Choose} $M$ a random element of $\ker \Lambda_{L^*,L}$.
    \item \textbf{While} $\mathrm{GCRD}(M,L)\neq 1$ \textbf{choose} $M$ to be another randomly chosen element of $\ker
          \Lambda_{L^*,L}$
    \item \textbf{For} $i$ \textbf{in} $\{1,\dots,k\}$ \textbf{add} $\mathrm{GCRD}(L,L^*_i M)$ to
          \emph{LCLM\_decomp}
    \item \textbf{Return} \emph{LCLM\_decomp}.
  \end{enumerate}
  \caption{LCLM\_dec}
  \label{algo_LCLMdec}
\end{algo}
We now present our $\mathrm{LCLM}$-decomposition algorithm in
Algorithm~\ref{algo_LCLMdec}. Before proving that it gives the correct
result we prove the following lemma:
\begin{lemma}\label{lem_centralelim}
  Let $L\in\Fbb_q(t)\langle\partial\rangle$ satisfy
  Hypothesis~\ref{hypo_sep}. Let $P_1|\dots|P_m$ be
  the Frobenius invariants of $\psi^{L}_p$ in $\Fbb_q(t^p)[Y]$. Let
  $N\in\Fbb_q(t^p)[Y]$ be separable and irreducible (over $\Fbb_q(t^p)$), and $\nu\in\mathbb{N}$ be maximal 
  such that $N(\partial^p)^\nu|L$. We denote by $\nu_N$ the valuation associated
      to $N$ over $\Fbb_q(t^p)[Y]$.
  \begin{enumerate}[label=\roman*)]
    \item $m\leq p$.
    \item  If $\nu>0$ then $m=p$ and $\nu_N(P_1)=\nu$.
    \item If $\nu_N(P_p)=\nu$ then $L=\mathrm{LCLM}(L\cdot
      N(\partial^p)^{-\nu},N(\partial^p)^\nu)$
  \end{enumerate}
\end{lemma}
\begin{proof}
  Since the Frobenius invariants of $\psi^L_p$ do not depend on the base
  field and any central simple algebra of the form $\Dcal_{N'(\partial^p)}$ must split over some finite
  dimensional Galois extension $K$ of $\Fbb_q(t)$,
  we may assume without loss of generality that $L$ has no central
  irreducible divisor (in particular we can assume $N$ to be to be irreducible as a
  polynomial and $N(\partial^p)$ to be reducible as an operator). We denote
  $C$ the constant field of $K$.
  Let $N'\in C[Y]$ be an irreducible factor of $P_1$. Then $N'$
  must appear in each of the $P_i$. Since we assumed $N'(\partial^p)$ to be
  reducible (otherwise it would be a central irreducible divisor of $L$), this means according to 
  Theorem~\ref{thm_decompclass} that a
  $\mathrm{LCLM}$-decomposition of $L$ must contain exactly $m$
  indecomposable factors which are divisors of a power of $N'(\partial^p)$.
  In particular, if $L_1,\dots,L_m$ are their unique right irreducible
  factors, then each $L_i$ is a divisor of $N'(\partial^p)$, thus
  $\mathrm{LCLM}_{i=1}^m L_i|N'(\partial^p)$, and
  $\ord \mathrm{LCLM}_{i=1}^m L_i=\sum_{i=1}^m\ord L_i=m\deg(N')\leq\ord
  N'(\partial^p)=p\deg(N')$ which proves $m\leq p$.\\
  The $K\langle\partial\rangle$-module $\Dcal_{N(\partial^p)^\nu}$
  can be identified as the submodule of $\Dcal_L$, $\ker N^\nu(\psi^L_p)$.
  $(\Dcal_{L},\psi^{L}_p)$ is isomorphic to
  $\bigoplus_{i=1}^m K[T]/P_i(T)$ as a $K[T]$-module. Thus
  $\ker
  N(\psi^L_p)\simeq \bigoplus_{i=1}^m K[T]/N^{\min(\nu,\nu_N(P_i))}(T)$ as
  $K[T]$-modules
  and \small$$\dim_{K}\ker N(\psi^L_p)=\sum_{i=1}^m
  \min(\nu,\nu_N(P_i))\deg(N)=p\nu\deg(N)$$\normalsize. This can only happen if
  $\nu=0$ or $m=p$ and $\nu(P_i)\geq \nu$ pour tout $i$. In particular
  $\nu_N(P_1)\geq \nu$. To prove that we have the equality we can write
  $L=L'N(\partial^p)^\nu$. Then $(\Dcal_{L'},\psi^{L'}_p)$ can be
  identified as the quotient of $\Dcal_L$ by $\ker N^\nu(\psi^L_p)$ and is
  thus isomorphic as a $K[T]$-module to $\bigoplus_{i=1}^p K[T]/P_i'$ where
  $P_i'=P_i\cdot N^{-\nu}$. In particular the $P_i'$ are the Frobenius
  invariants of $\psi^{L'}_p$. If $\nu_N(P_1')>0$ then we
  can show the same way that we proved that $m\leq p$ that
  $N(\partial^p)|L'$ which is a contradiction with the maximality of
  $\nu$. To prove (iii) we may assume that $\nu>0$. Since $N(\partial^p)$
  is central, $L$ is indeed a left multiple of $L'=L\cdot
  N(\partial^p)^{-\nu}$ and $N(\partial^p)^{\nu}$. Furthermore by hypothesis,
  $\nu_N(\chi(\psi^L_p))=p\nu$. It follows by multiplicativity that
  $\nu_N(\chi(\psi^{L'}_p))=0$, thus $L'$ and $N(\partial^p)^\nu$ are
  coprime. By Proposition~\ref{prop_lclm_directsum}(i) it follows that
  $L=\mathrm{LCLM}(L',N(\partial^p)^\nu)$.
\end{proof}
\begin{theorem}
  Algorithm~\ref{algo_LCLMdec} is correct. Furthermore if
  $L\in\Fbb_q(t)\langle\partial\rangle$ is of order $r$ and degree $d$
  (see Definition~\ref{def_degree}), then Algorithm~\ref{algo_LCLMdec}
  terminates in time polynomial in $r,d$ and $p$ and yields a
  $\mathrm{LCLM}$-decomposition of $L$ where each factor has degree
  polynomial in $r$ and $d$ and at most quadratic in $p$.
\end{theorem}
\begin{remark}
  The quadratic dependency in $p$ comes from step(6) of
  Algorithm~\ref{algo_LCLMdec}. However, experiments suggest that the
  coefficients of elements of the basis of $\ker \Lambda_{L^*,L}$ are of no bigger size
  than the coefficients of the matrix of $\Lambda_{L_*,L}$ itself. More
  accuratly, the growth of the coefficients does not seem to depend on $p$ and
  rather only on $r$ and $d$. This suggests that the output of
  Algorithm~\ref{algo_LCLMdec} is actually of size quasilinear and not
  quasiquadratic
  in $p$.
\end{remark}
\begin{proof}
  At the end of step (4), $L$ can no longer have any central irreducible
  divisor and the $Q_i$ are such that if $P_i(Y^p)=Q_i^p(Y)$ then the $P_i$
  are the Frobenius invariants of its $p$-curvature. This is a direct
  consequence of Lemma~\ref{lem_centralelim}. According to (iii) of the
  same lemma, to compute a $\mathrm{LCLM}$-decomposition of $L$ (in the
  input), it is enough to compute a $\mathrm{LCLM}$-decomposition of $L$ in
  step $4$ and of $N^p(\partial)$ when $\nu_N(Q_1)=\nu_N(Q_m)$ and $m=p$.
  The fact that Algorithm~\ref{algo_nicerepr} computes such a decomposition
  in the latter case when $N^p(\partial)$ is reducible is a consequence of
  Proposition~\ref{prop_buildrepr} since $N^{p\nu}(\partial)$ is a common
  multiple of the right order.\\
  In step (5), Algorithm~\ref{algo_nicerepr} yields an operator $L^*$ equivalent
  to $L$ in step (4). Indeed, by Theorem~\ref{thm_correctalgo1},
  $\psi^{L^*}_p$ and $\psi^{L}_p$ have the same Frobenius invariants and
  are thus equivalent, which means that $L^*$ and $L$ are equivalent by
  Corollary~\ref{cor_critiso}. The elements in the kernel of $\Lambda_{L_*,L}$ are
  all the morphisms of $\Fbb_q(t)\langle\partial\rangle$-modules between
  $\Dcal_{L^*}$ and $\Dcal_L$. Selecting one coprime with $L$ ensures that
  it is an isomorphism and the
  fact that $L^*$ and $L$ are equivalent ensures the existence of such an
  element. We conclude by Lemma~\ref{lm_decpropag}.\\
  Since $\sum_{i=1}^m \deg_t(Q_i)\leq p^{-1}\deg_t(\chi(\psi^L_p))\leq d$
  (see~\cite[Lemma~3.9]{BoCaSc14}) and $\sum_{i=1}^m
  \deg_Y(Q_i)=r$, it follows from Theorem~\ref{thm_correctalgo1} that $L^*$
  and the $L_i^*$ all have degree polynomial in $r$ and $d$. The matrix of
  $\Lambda_{L^*,L}$ in the bases $(x^i\partial^j)_{0\leq i<p,0\leq
  j<r}$ can thus be represented by a matrix in
  $M_{rp}(\Fbb_q(t^p))$ whose coefficients are of degree in $t^p$ polynomial in
  $r$ and $d$. It follows that elements of the kernel have coefficients of
  degree in $t^p$ polynomial in $r$ and $d$ and linear in $p$. Thus $M$ has
  is of degree in $t$ polynomial in $r$ and $d$ and quadratic in $p$.
\end{proof}

\bibliographystyle{acm}
\bibliography{../bibliographie}

\begin{thebibliography}{10}

\bibitem{Anderson1974}
{\sc Anderson, F.~W., and Fuller, K.~R.}
\newblock Rings and categories of modules.

\bibitem{BJN94}
{\sc Bhattacharya, P.~B., Jain, S.~K., and Nagpaul, S.}
\newblock {\em Basic abstract algebra}.
\newblock Cambridge University Press, 1994.

\bibitem{BoCaSc14}
{\sc Bostan, A., Caruso, X., and Schost, E.}
\newblock A fast algorithm for computing the characteristic polynomial of the
  $p$-curvature.
\newblock In {\em I{SSAC} 2014---{P}roceedings of the 39th {I}nternational
  {S}ymposium on {S}ymbolic and {A}lgebraic {C}omputation\/} (2014), ACM, New
  York, pp.~59--66.

\bibitem{BoCaSc16}
{\sc Bostan, A., Caruso, X., and Schost, E.}
\newblock Computation of the similarity class of the {$p$}-curvature.
\newblock In {\em Proceedings of the 2016 {ACM} {I}nternational {S}ymposium on
  {S}ymbolic and {A}lgebraic {C}omputation\/} (2016), ACM, New York,
  pp.~111--118.

\bibitem{BCSZ12}
{\sc Bostan, A., Chyzak, F., Salvy, B., and Li, Z.}
\newblock Fast computation of common left multiples of linear ordinary
  differential operators.
\newblock In {\em Proceedings of the 37th International Symposium on Symbolic
  and Algebraic Computation\/} (July 2012), ISSAC’12, ACM.

\bibitem{ChGoMe22}
{\sc Chyzak, F., Goyer, A., and Mezzarobba, M.}
\newblock Symbolic-numeric factorization of differential operators.
\newblock In {\em Proceedings of the 2022 International Symposium on Symbolic
  and Algebraic Computation\/} (New York, NY, USA, 2022), ISSAC '22,
  Association for Computing Machinery, p.~73–82.

\bibitem{Cluzeau03}
{\sc Cluzeau, T.}
\newblock Factorization of differential systems in characteristic {$p$}.
\newblock In {\em Proceedings of the 2003 {I}nternational {S}ymposium on
  {S}ymbolic and {A}lgebraic {C}omputation\/} (2003), ACM, New York,
  pp.~58--65.

\bibitem{Dw90}
{\sc Dwork, B.}
\newblock Differential operators with nilpotent {$p$}-curvature.
\newblock {\em Amer. J. Math. 112}, 5 (1990), 749--786.

\bibitem{Eid21}
{\sc Eid, E.}
\newblock Fast computation of hyperelliptic curve isogenies in odd
  characteristic.
\newblock In {\em Proceedings of the 2021 on International Symposium on
  Symbolic and Algebraic Computation\/} (New York, NY, USA, 2021), ISSAC '21,
  Association for Computing Machinery, p.~131–138.

\bibitem{Gai25}
{\sc Gaillard, L.}
\newblock A unified approach for degree bound estimates of linear differential
  operators.
\newblock In {\em Proceedings of the 2025 International Symposium on Symbolic
  and Algebraic Computation\/} (New York, NY, USA, 2025), ISSAC '25,
  Association for Computing Machinery, p.~16–24.

\bibitem{GiZh03}
{\sc Giesbrecht, M., and Zhang, Y.}
\newblock Factoring and decomposing {O}re polynomials over {$\Bbb F_q(t)$}.
\newblock In {\em Proceedings of the 2003 {I}nternational {S}ymposium on
  {S}ymbolic and {A}lgebraic {C}omputation\/} (2003), ACM, New York,
  pp.~127--134.

\bibitem{GiSz06}
{\sc Gille, P., and Szamuely, T.}
\newblock {\em Central Simple Algebras and Galois Cohomology}.
\newblock Cambridge Studies in Advanced Mathematics. Cambridge University
  Press, 2006.

\bibitem{GoTLoNa19}
{\sc G\'{o}mez-Torrecillas, J., Lobillo, F.~J., and Navarro, G.}
\newblock Computing the bound of an {O}re polynomial. {A}pplications to
  factorization.
\newblock {\em J. Symbolic Comput. 92\/} (2019), 269--297.

\bibitem{Gri90}
{\sc Grigor'ev, D.}
\newblock Complexity of factoring and calculating the gcd of linear ordinary
  differential operators.
\newblock {\em Journal of Symbolic Computation 10}, 1 (1990), 7--37.

\bibitem{kauers23c}
{\sc Kauers, M.}
\newblock {\em D-Finite Functions}.
\newblock Springer, 2023.

\bibitem{LaVa16}
{\sc Lairez, P., and Vaccon, T.}
\newblock On {$p$}-adic differential equations with separation of variables.
\newblock In {\em Proceedings of the 2016 {ACM} {I}nternational {S}ymposium on
  {S}ymbolic and {A}lgebraic {C}omputation\/} (2016), ACM, New York,
  pp.~319--323.

\bibitem{Lauder04count}
{\sc Lauder, A. G.~B.}
\newblock Counting solutions to equations in many variables over finite fields.
\newblock {\em Found. Comput. Math. 4}, 3 (2004), 221--267.

\bibitem{PagPhD}
{\sc Pag{\`e}s, R.}
\newblock {\em {Factoring differential operators in positive characteristic.}}
\newblock Theses, {Universit{\'e} de Bordeaux}, Feb. 2024.

\bibitem{Pag24}
{\sc Pag{\`e}s, R.}
\newblock {Solving the p-Riccati Equations and Applications to the
  Factorisation of Differential Operators.}
\newblock working paper or preprint, Jan. 2024.

\bibitem{vdH07:galois}
{\sc van~der Hoeven, J.}
\newblock Around the numeric-symbolic computation of differential {Galois}
  groups.
\newblock {\em JSC 42\/} (2007), 236--264.

\bibitem{vdH07:reshol}
{\sc van~der Hoeven, J.}
\newblock Efficient accelero-summation of holonomic functions.
\newblock {\em JSC 42}, 4 (2007), 389--428.

\bibitem{Put95}
{\sc van~der Put, M.}
\newblock Differential equations in characteristic {$p$}.
\newblock vol.~97. 1995, pp.~227--251.
\newblock Special issue in honour of Frans Oort.

\bibitem{Put96}
{\sc van~der Put, M.}
\newblock Reduction modulo {$p$} of differential equations.
\newblock {\em Indag. Math. (N.S.) 7}, 3 (1996), 367--387.

\bibitem{PuSi03}
{\sc van~der Put, M., and Singer, M.~F.}
\newblock {\em Galois theory of linear differential equations}, vol.~328 of
  {\em Grundlehren der Mathematischen Wissenschaften [Fundamental Principles of
  Mathematical Sciences]}.
\newblock Springer-Verlag, Berlin, 2003.

\bibitem{Hoeij96}
{\sc van Hoeij, M.}
\newblock Rational solutions of the mixed differential equation and its
  application to factorization of differential operators.
\newblock In {\em Proceedings of the 1996 International Symposium on Symbolic
  and Algebraic Computation\/} (New York, NY, USA, 1996), ISSAC '96,
  Association for Computing Machinery, p.~219–225.

\bibitem{VHo97}
{\sc {Van Hoeij}, M.}
\newblock Factorization of differential operators with rational functions
  coefficients.
\newblock {\em Journal of Symbolic Computation 24}, 5 (1997), 537--561.

\end{thebibliography}

\end{document}